\documentclass[conference]{IEEEtran}
\IEEEoverridecommandlockouts
\usepackage{amsmath,graphicx}
\usepackage{amsmath,amssymb,amsfonts}
\usepackage{cite, soul}
\usepackage{url}
\usepackage{amsthm}
\usepackage{bm}
\usepackage{comment}
\usepackage{multirow}
\usepackage{booktabs}
\usepackage{hhline}
\usepackage{graphicx}     
\usepackage{caption}      
\usepackage{subcaption}   

\usepackage{setspace}
\usepackage{setspace}
\usepackage[ruled,linesnumbered]{algorithm2e}

\usepackage[usenames,dvipsnames]{color}

\newtheorem{lem}{Lemma}

\def\@#1{\pmb{#1}}
\def\b#1{\mathbb{#1}}
\def\bf#1{\mathbf{#1}}
\def\s#1{\mathsf{#1}}

\def\ca#1{\mathcal{#1}}

\title{Distributed On-Device LLM Inference With Over-the-Air Computation}
\author{Kai Zhang, Hengtao He, Shenghui Song, Jun Zhang, \emph{Fellow}, \emph{IEEE}, and Khaled B. Letaief, \emph{Fellow}, \emph{IEEE}\\
Department of Electronic and Computer Engineering, The Hong Kong University of Science and Technology.\\
Email: kzhangbn@connect.ust.hk, \{eehthe, eeshsong, eejzhang, eekhaled\}@ust.hk
}

\begin{document}


\maketitle
\begin{abstract}
	
Large language models (LLMs) have achieved remarkable success across various artificial intelligence tasks. However, their enormous sizes and computational demands pose significant challenges for the deployment on edge devices. To address this issue, we present a distributed on-device LLM inference framework based on tensor parallelism, which partitions neural network tensors (e.g., weight matrices) of LLMs among multiple edge devices for collaborative inference. Nevertheless, tensor parallelism involves frequent all-reduce operations to aggregate intermediate layer outputs across participating devices during inference, resulting in substantial communication overhead. To mitigate this bottleneck, we propose an over-the-air computation method that leverages the analog superposition property of wireless multiple-access channels to facilitate fast all-reduce operations. To minimize the average transmission mean-squared error, we investigate joint model assignment and transceiver optimization, which can be formulated as a mixed-timescale stochastic non-convex optimization problem. Then, we develop a mixed-timescale algorithm leveraging semidefinite relaxation and stochastic successive convex approximation methods. Comprehensive simulation results will show that the proposed approach significantly reduces inference latency while improving accuracy. This makes distributed on-device LLM inference practical for resource-constrained edge devices.

\end{abstract}

\begin{IEEEkeywords} 6G, distributed inference, large language models, over-the-air computation, tensor parallelism. \end{IEEEkeywords}

\begin{figure*}[t]
	\renewcommand\figurename{\small Fig.}
	\centering \setlength{\baselineskip}{2pt}
	\includegraphics[width = 1\textwidth,trim=0 197 0 113,clip]{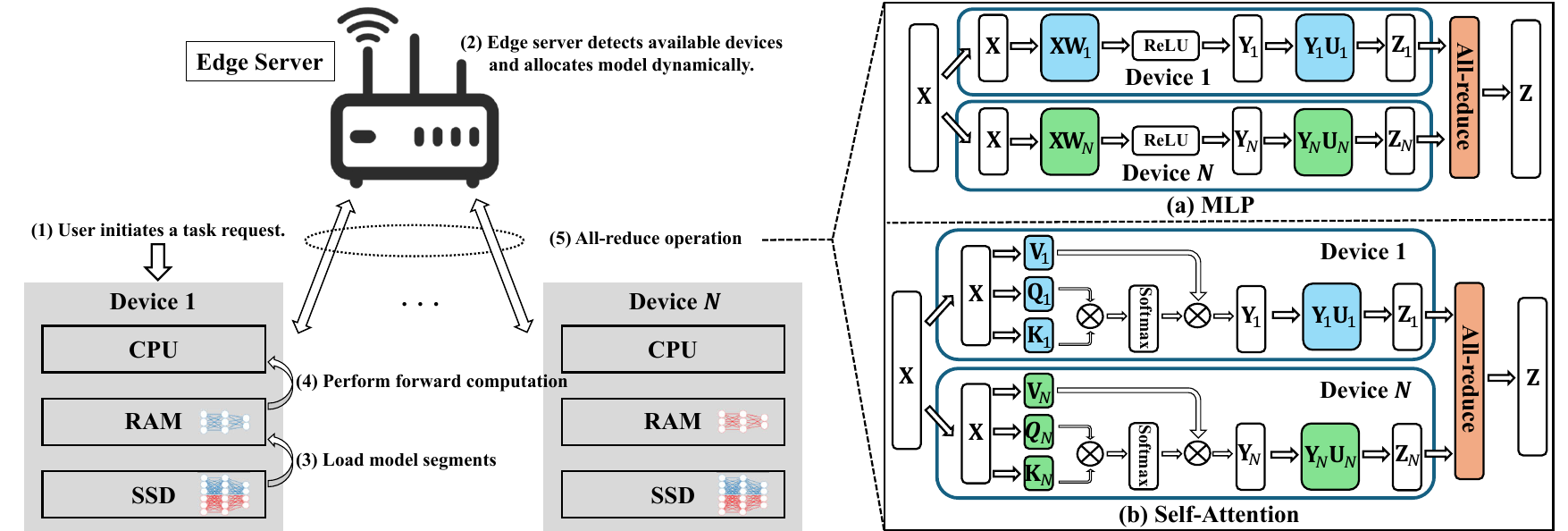}
	\vspace{-2pt}
	\caption{An illustration of the distributed on-device LLM inference system, showing the system workflow and visualizing tensor parallelism for (a) MLP and (b) self-attention layers.}\label{illustration}
	\vspace{-5pt}
\end{figure*}

\section{Introduction}

Large language models (LLMs) have achieved remarkable success in various fields of artifical intelligence (AI), such as natural language processing \cite{min2023recent} and embodied intelligence \cite{fan2024embodied}.
However, most existing LLMs rely on cloud-based infrastructure due to their enormous computational and memory requirements, which pose significant challenges for the deployment on edge devices. This cloud-based approach raises critical concerns regarding scalability and privacy, particularly in sensitive fields like healthcare and finance.
To address these limitations, distributed LLM inference has recently been proposed as a promising solution, which distributes inference workloads across multiple devices \cite{wu2023fast}. This strategy reduces the burden on the individual device and strengthens privacy protections.
Furthermore, the advanced communication capabilities of 5G and future 6G wireless networks have made distributed LLM inference increasingly promising for real-time applications \cite{letaief2019roadmap,letaief2021edge}.

The performance of distributed LLM inference systems is largely determined by communication overhead.
To enhance the communication efficiency of distributed LLM inference, several recent studies have been conducted \cite{zhang2024edgeshard, he2024large, chen2024adaptive, shao2021learning,li2024tackling}.
In \cite{zhang2024edgeshard}, the authors proposed a collaborative edge computing method for distributing different layers of LLMs across the edge device and cloud server, and developed a device selection and model partitioning algorithm to reduce inference latency and optimize the throughput.
In \cite{he2024large}, an active inference-based method was proposed to address the task offloading and resource allocation problem for LLM inference in cloud-edge computing frameworks.
Similarly, the authors of \cite{chen2024adaptive} proposed a reinforcement learning algorithm that optimizes the splitting point of LLMs between the edge device and cloud server to reduce the communication overhead.
Furthermore, task-oriented communications have been utilized to optimize end-to-end inference throughput, accuracy, and latency, which can further enhance the communication efficiency of distributed LLM inference \cite{shao2021learning,li2024tackling}.

Despite significant advances in distributed LLM inference, existing works \cite{zhang2024edgeshard, chen2024adaptive,he2024large, shao2021learning,li2024tackling} primarily focus on the device-cloud collaborative inference, which faces scalability challenges due to the reliance on the centralized cloud server. To address this limitation, distributed on-device LLM inference with tensor parallelism has recently been proposed \cite{shoeybi2019megatron}. This approach divides neural network tensors (e.g., weight matrices) of LLMs into smaller segments and distributes them across multiple edge devices. However, a critical challenge in tensor parallelism is the frequent all-reduce operations required to aggregate intermediate outputs across devices, which can cause substantial latency in practical wireless networks and hinder the real-time inference.

To solve the problem, we present a communication-efficient framework for distributed on-device LLM inference in this paper.
By exploiting the analog superposition property of wireless multiple-access channels, we propose an over-the-air computation method to facilitate fast all-reduce operations in tensor parallelism.
The performance of the distributed LLM inference is highly determined by the transmission error during the over-the-air computation.
To minimize the average transmission mean-squared error (MSE) with limited energy supply for edge devices, we investigate a joint model assignment and transceiver design problem, which can be formulated as a mixed-timescale stochastic non-convex optimization.
Specifically, the model assignment is determined at the beginning of inference based on long-term statistical channel state information (CSI), while the transceiver design adapts dynamically to the CSI.
To solve the problem, we develop a mixed-timescale algorithm leveraging semidefinite relaxation (SDR) and stochastic successive convex approximation (SCA) methods.
Extensive simulation results demonstrate that the proposed approach effectively reduces inference latency and improves inference accuracy.

\emph{Notations:} Column vectors and matrices are denoted by boldface lowercase and boldface capital letters, respectively. $\b{C}^{M\times N}$ represents the space of the $M \times N$ complex-valued matrices. $(\cdot)^\s{T}$ and $(\cdot)^\s{H}$ stand for the transpose and the conjugate transpose of their arguments, respectively. $\textup{tr}(\bf{A})$ denote the trace of matrix $\bf{A}$.
$\b{E}[\cdot]$ denotes the expectation operation. $\nabla$ represents the gradient operator.

\section{System Model and Problem Formulation}

In this section, we first elaborate on the distributed on-device LLM inference system, followed by the proposal of an over-the-air computation approach to accelerate all-reduce operations during inference. Subsequently, we formulate a joint model assignment and transceiver design problem to minimize the average transmission MSE.

\subsection{Proposed Distributed On-Device LLM Inference System}

To deploy LLMs on resource-limited edge devices, distributed on-device inference with tensor parallelism has been proposed.
This method involves partitioning neural network tensors (e.g., weight matrices) of LLMs into smaller segments and distributing them across multiple devices for simultaneous processing.
The complete workflow of the distributed on-device LLM inference system is illustrated in Fig. \ref{illustration}.
When a device initiates an inference request, the edge server dynamically identifies available local devices and partitions the model parameters. Then, each device loads its assigned model segment into memory and performs forward computation. After each layer of LLM is computed, an all-reduce operation aggregates the intermediate outputs from all devices, ensuring synchronization and consistency across devices during inference.

LLMs are primarily built on the Transformer architecture, which typically consists of dozens of Transformer layers. Each Transformer layer includes a self-attention mechanism and a multi-layer perceptron (MLP).
To achieve efficient distributed inference, tensor parallelism partitions both the self-attention layer and the MLP layer of each Transformer block into smaller tensor segments, as shown in Fig. \ref{illustration}.
For a typical 2-layer MLP within the Transformer block, the forward computation involves two main linear transformations, separated by a non-linear activation function (e.g., ReLU or GeLU). Mathematically, it can be expressed as follows,
\begin{equation}
	\begin{aligned}
		\bf{Z} = \max(\bf{0}, \bf{XW})\bf{U},
	\end{aligned}    
\end{equation}
where $\bf{X}$ is the input to the MLP layer, $\bf{Z}$ is the output, and $\bf{W}$ and $\bf{U}$ are the weight matrices, respectively.
The traditional centralized inference involves loading the entire weight matrices $\bf{W}$ and $\bf{U}$ into memory and performing full matrix multiplications on a single device, which is usually impractical for resource-limited edge devices.
To overcome this challenge, tensor parallelism distributes the weight matrices $\bf{W}$ and $\bf{U}$ across $N$ devices, as illustrated below,
\begin{equation}
	\begin{aligned}
		\bf{W} &= [\bf{W}_1, \ldots, \bf{W}_N],\\
		\bf{U} &= [\bf{U}_1^{\s{T}}, \ldots, \bf{U}_N^{\s{T}}]^{\s{T}},
	\end{aligned}    
\end{equation}
where $\bf{W}_n$ and $\bf{U}_n$ represent the portions of the weight matrices assigned to device $n$. Then, each device performs the forward computation on its respective segment. The input tensor $\bf{X}$ for the MLP layer is first broadcasted to all devices, and each device $n$ computes its local contribution as follows,
\begin{equation}
	\begin{aligned}
		\bf{Z}_n = \mathrm{max}(0, \bf{X} \bf{W}_n) \bf{U}_n,
	\end{aligned}    
\end{equation}
where $\bf{Z}_n$ is the partial output produced by device $n$. 
Once all devices obtain their local outputs $\bf{Z}_n$, an all-reduce operation is performed to aggregate the partial outputs from all devices as follows,
\begin{equation}\label{aggre_rule}
	\begin{aligned}
		\bf{Z} = \sum_{n=1}^N \bf{Z}_n.
	\end{aligned}    
\end{equation}
After aggregation, the final output $\bf{Z}$ of the MLP layer is broadcasted to all devices, ensuring synchronization and consistency across devices for the subsequent layer's computation.

For the self-attention layer, tensor parallelism similarly partitions the query ($\bf{Q}$), key ($\bf{K}$), and value ($\bf{V}$) matrices across devices.
Since the computation in the self-attention layer shares a similar matrix-multiplication structure to the MLP layer, the partitioning, local computation, and aggregation steps follow the same principles. A similar all-reduce operation is required to gather and combine the partial outputs from devices.

\subsection{Over-the-Air All-Reduce}
Employing tensor parallelism for distributed inference requires frequent all-reduce operations, which causes significant latency in practical wireless networks.
To address this issue, we propose an over-the-air computation approach to accelerate the all-reduce steps. Over-the-air computation aggregates distributed data efficiently by leveraging signal superposition in multiple-access channels, allowing simultaneous transmissions to compute nomographic functions (e.g., arithmetic mean) \cite{goldenbaum2013harnessing}. In the distributed LLM inference system, the aggregation of partial intermediate outputs across devices aligns with this operation, making over-the-air computation suitable to mitigate the communication overhead.

Specifically, we assume the edge server and each edge device are equipped with $N_r$ and $N_t$ antennas, respectively, constructing a MIMO multiple-access channel.
We consider the block-fading channel model, where channel statistics remain constant throughout the inference process, with channel states varying independently across different time intervals.
Let $\bf{s}_n=[s_{n,1}, \ldots, s_{n,L}]^{\s{T}}$ denote the per-round transmitted $L$ entries of device $n$'s intermediate output, where the complete intermediate output has a dimensionality of $L_0$.
Given synchronized symbol boundaries, all devices transmit their intermediate outputs simultaneously. 
To mitigate the distortion of received signals caused by channel noise at the server, aggregation beamforming is applied.
Let $\mathbf{A} \in \mathbb{C}^{N_r \times L}$ and $\mathbf{B}_n \in \mathbb{C}^{N_t \times L}$ denote the aggregation beamforming matrix at the edge server and the data precoding matrix at device $n$, respectively.
Then, the received signal at the server after the over-the-air computation is given by,
\begin{equation}
	\begin{aligned}
		&\hat{\bf{s}}  = \bf{A}^{\s{H}}\sum_{n=1}^{N}\mathbf{H}_n \mathbf{B}_n \bf{s}_n + \bf{A}^{\s{H}} \bf{n},
	\end{aligned}    
\end{equation}
where $\bf{H}_n\in \b{C}^{N_r\times N_t}$ denotes the uplink MIMO channel from device $n$ to the edge server, and $\bf{n}\sim \mathcal{CN}\left(0, \sigma_{z}^{2} \mathbf{I}\right)$ denotes the additive white Gaussian noise vector with $\sigma_{z}^{2}$ being the noise power.
The distortion of $\hat{\bf{s}}$ with respect to the desired target vector $\bf{s}  = \sum_{n=1}^{N} \bf{s}_n$ is measured by the MSE, defined as
%
%
%
%
%
%
%
\begin{equation}
	\begin{aligned}
		\textup{MSE}(\hat{\bf{s}},\bf{s}) = \mathbb{E}\left[(\hat{\bf{s}}-\bf{s})^{\s{T}}(\hat{\bf{s}}-\bf{s})\right].
	\end{aligned}    
\end{equation}
The MSE serves as a metric to evaluate the performance of the all-reduce operation with over-the-air computation. By substituting (5) into (6), the MSE can be explicitly represented as a function of transceiver beamforming matrices as follows,
\begin{equation}
	\begin{aligned}
		&\textup{MSE}(\mathbf{A},\left\lbrace \mathbf{B}_n  \right\rbrace ) \\
		&\!= \sum_{n=1}^{N} \textup{tr}\! \left( \! \left( \mathbf{A}^{\!\s{H}}  \mathbf{H}_n \mathbf{B}_n - \mathbf{I} \right) \! \left( \mathbf{A}^{\!\s{H}}  \mathbf{H}_n \mathbf{B}_n - \mathbf{I} \right)^{\s{H}} \right) \!+\! \sigma_z^2 \textup{tr}\left( \mathbf{A}^{\!\s{H}} \mathbf{A} \right).   \\
	\end{aligned}    
\end{equation}

Edge devices involved in inference tasks typically have limited energy supply. Thus, we assume that for each device $n$, the energy consumption for both the forward computation of each LLM layer and the transmission of the intermediate output cannot exceed the maximum power budget $P_{n}^{\textup{max}}$.
To model the computation energy consumption, we first introduce a model assignment vector $\bf{m}=[m_1,\ldots,m_N]$, which enables more flexible and efficient model distribution based on device capabilities (e.g., memory size and compute power). The entry $m_n \in [0,1]$ of $\bf{m}$ represents the proportion of the model allocated to device $n$. Consequently, the computation energy consumption for device $n$ is given by $e_n m_n s^{\textup{tot}}$, where $e_n$ denotes the device-specific energy coefficient that reflects the energy cost associated with accessing and processing each weight during computation, and $s^{\textup{tot}}$ is the number of parameters (weights) for each layer. The communication energy consumption of device $n$ can be derived as $\frac{L_0}{L} \textup{tr}\left( \mathbf{B}_n \mathbf{B}_n^{\s{H}} \right)$. Accordingly, the energy constraint is given by
\begin{equation}
	\begin{aligned}
		e_n m_n s^{\textup{tot}} + \frac{L_0}{L} \textup{tr}\left( \mathbf{B}_n \mathbf{B}_n^{\s{H}} \right) \leq P_{n}^{\textup{max}} , \forall n.
	\end{aligned}    
\end{equation}

\subsection{Problem Formulation}

In the proposed distributed LLM inference system, the overall performance depends on the model assignment policy $\bf{m}$ and the transceiver beamformers, $\mathbf{A}$ and $\left\lbrace \mathbf{B}_n \right\rbrace$. The transceiver optimization focuses on minimizing signal misalignment errors while suppressing noise, and thus it adapts dynamically to instantaneous CSI, making it a fast-timescale variable.
However, real-time adaptation of model assignment to instantaneous CSI is impractical due to the high latency caused by loading model segments. Thus, model assignment should be determined at the start of the inference process and should depend on long-term channel statistics, functioning as a slow-timescale decision. The resulting problem is a mixed-timescale joint optimization of short-term transceiver beamformers $\mathbf{A}$, $\left\lbrace \mathbf{B}_n \right\rbrace$ and long-term model assignment $\bf{m}$, with the goal of minimizing the average MSE as follows,
\begin{equation}
	\begin{aligned}
		\ca{P}:~\min_{\bf{m}} &~ \mathbb{E}_{\mathbf{H}}\left[ \min_{\mathbf{A},\left\lbrace \mathbf{B}_n  \right\rbrace} \textup{MSE}(\mathbf{A},\left\lbrace \mathbf{B}_n  \right\rbrace )    \right]  \\
		\textup{s.t.}& ~ e_n m_n s^{\textup{tot}} + \frac{L_0}{L} \textup{tr}\left( \mathbf{B}_n \mathbf{B}_n^{\s{H}} \right) \leq P_{n}^{\textup{max}}, \forall n,\\
		&~ \bf{m}^{\textsf{T}}\boldsymbol{1}=1, \bf{m}\geq 0,
	\end{aligned}    
\end{equation}
where the expectation $\mathbb{E}_{\mathbf{H}}\left[ \cdot \right] $ is taken over all random channel realizations $\mathbf{H} = \left\lbrace \bf{H}_n \right\rbrace_{n=1}^N $. 
Problem $\ca{P}$ is challenging to solve for three reasons: 1) the inherent non-convexity caused by the coupling between transceiver beamformers, 2) the presence of expectations over the random channel state in the objective function, and 3) the interdependence between the short-term beamforming matrices and the long-term model assignment policy within the per-device power constraints.
To address these issues, we develop an efficient mixed-timescale algorithm in the following section.

\section{Algorithm Development}
In this section, we develop a mixed-timescale algorithm to solve the joint model assignment and transceiver optimization problem $\ca{P}$.
We start by decomposing problem $\ca{P}$ into a family of short-term subproblems and a long-term subproblem as follows.

\subsubsection{Short-term transceiver optimization for given model assignment policy $\bf{m}$ and channel condition $\bf{H}$}
\begin{equation}
	\begin{aligned}
		\ca{P}_{s}: \min_{\mathbf{A},\left\lbrace \mathbf{B}_n  \right\rbrace}  &~ \textup{MSE}(\mathbf{A},\left\lbrace \mathbf{B}_n  \right\rbrace )   \\
		\textup{s.t.}~~& ~ e_n m_n s^{\textup{tot}} + \frac{L_0}{L} \textup{tr}\left( \mathbf{B}_n \mathbf{B}_n^{\s{H}} \right) \leq P_{n}^{\textup{max}}, \forall n.
	\end{aligned}    
\end{equation}

\subsubsection{Long-term model assignment optimization based on the optimal solution $\mathbf{A}^{*}(\bf{m}),\left\lbrace \mathbf{B}_n^{*}(\bf{m})  \right\rbrace$ to problem $\ca{P}_{s}$}

	\begin{align}
		\ca{P}_l:\min_{\bf{m}} &~ \mathbb{E}_{\mathbf{H}}\left[ \textup{MSE}(\mathbf{A}^{*}(\bf{m}),\left\lbrace \mathbf{B}_n^{*}(\bf{m})  \right\rbrace )    \right]  \nonumber  \\
		\textup{s.t.}& ~ e_n m_n s^{\textup{tot}} + \frac{L_0}{L} \textup{tr}\left( \mathbf{B}_n^{*}(\bf{m}) \mathbf{B}_n^{*}(\bf{m})^{\s{H}} \right) \leq P_{n}^{\textup{max}}, \forall n, \nonumber\\
		&~ \bf{m}^{\textsf{T}}\boldsymbol{1}=1, \bf{m}\geq 0.
	\end{align}    

The short-term problem $\ca{P}_s$ remains non-convex, which we address using the SDR technique. The long-term model assignment problem $\ca{P}_l$ is similarly challenging, as the optimal $\mathbf{A}^{*}(\bf{m}),\left\lbrace \mathbf{B}_n^{*}(\bf{m})  \right\rbrace$ cannot be derived in closed form. Additionally, the distribution of the channel state is difficult to obtain in practical wireless systems. To address these challenges, we propose a stochastic SCA algorithm that operates without requiring prior knowledge of the channel state distribution. In the following subsections, we provide a detailed implementation of the proposed algorithms.

\subsection{Short-Term Transceiver Optimization for $\ca{P}_s$}

We first simplify problem $\ca{P}_{s}$ by demonstrating that the zero-forcing (channel inversion) precoder is optimal conditioned on the aggregation beamformer.

\begin{lem}\label{lem_opt_precoder}
	For a given aggregation beamformer $\bf{A}$, the transmission MSE is minimized by using the zero-forcing precoders as follows,
	\begin{equation}\label{opt_precoder}
		\begin{aligned}
			\mathbf{B}_n^* =\left(  \mathbf{A}^{\!\s{H}}  \mathbf{H}_n\right)^{\!\s{H}} \left(  \mathbf{A}^{\!\s{H}}  \mathbf{H}_n \mathbf{H}_n^{\s{H}} \mathbf{A} \right)^{\! -1}, \forall n.
		\end{aligned}    
	\end{equation}
\end{lem}
\begin{proof}
	{Lemma \ref{lem_opt_precoder} can be proved by following the same steps as in \cite[Appendix A]{li2019wirelessly}. The detailed proof is omitted due to space limitation.}
\end{proof}

Let $\bf{G}$ represent the normalized aggregation beamformer that satisfies $\textup{tr}(\bf{G}\bf{G}^{\s{H}})=1$, and consequently $\bf{A}=\sqrt{\alpha} \bf{G}$ with $\alpha$ denoting the norm of $\bf{A}$.
By substituting \eqref{opt_precoder}, problem $\ca{P}_{s}$ can be reformulated as follows,
\begin{equation}\label{p_s_1}
	\begin{aligned}
		\min_{\alpha, \bf{G}} &~~ \alpha\\
		\textup{s.t.} & ~~ e_n m_n s^{\textup{tot}} +  \frac{L_0}{\alpha L} \textup{tr}\left( \left( \bf{G}^{\s{H}} \bf{H}_n \bf{H}_n^{\s{H}} \bf{G} \right)^{-1}  \right) \leq P_{n}^{\textup{max}}, \forall n,\\
		& ~~\textup{tr}\left( \bf{G} \bf{G}^{\s{H}}  \right) = 1.
	\end{aligned}  
\end{equation}
Problem \eqref{p_s_1} remains challenging to solve due to its non-convex constraints involving the term $\textup{tr}( ( \bf{G}^{\s{H}} \bf{H}_n \bf{H}_n^{\s{H}} \bf{G} )^{-1}  )$. To address the problem, we develop a tractable approximation of the problem by employing the following inequality,

\begin{equation}\label{sdr_ineq}
	\begin{aligned}
			 \textup{tr}\left( \left( \bf{G}^{\s{H}} \bf{H}_n \bf{H}_n^{\s{H}} \bf{G} \right)^{-1}  \right) \leq \frac{L}{\lambda_{\min}\left( \bf{H}_n^{\s{H}} \bf{G} \bf{G}^{\s{H}} \bf{H}_n \right)},
		\end{aligned}  
\end{equation}
where the equality holds when the channel is well-conditioned, i.e., the singular values of $\bf{H}_n$ are identical. By utilizing \eqref{sdr_ineq}, we reformulate an approximated version of problem \eqref{p_s_1} as follows,
\begin{equation}\label{p_s_2}
	\begin{aligned}
		\min_{\alpha, \bf{G}} &~~ \alpha\\
		\textup{s.t.} & ~~ \frac{L_0}{\alpha \lambda_{\min}\left( \bf{H}_n^{\s{H}} \bf{G} \bf{G}^{\s{H}} \bf{H}_n \right)} \leq P_{n}^{\textup{max}}-e_n m_n s^{\textup{tot}}, \forall n,\\
		& ~~\textup{tr}\left( \bf{G} \bf{G}^{\s{H}}  \right) = 1.
	\end{aligned}  
\end{equation}
Then, by introducing a new variable ${\hat{\bf G}}=\bf{G} \bf{G}^{\s{H}}$, an equivalent formulation of problem \eqref{p_s_2} is obtained as follows,
\begin{equation}\label{p_s_3}
	\begin{aligned}
		\min_{\alpha, \hat{\bf{G}}} &~~ \alpha\\
		\textup{s.t.} & ~~ \frac{L_0}{\alpha \lambda_{\min}\left( \bf{H}_n^{\s{H}} \hat{\bf{G}} \bf{H}_n \right)} \leq P_{n}^{\textup{max}}-e_n m_n s^{\textup{tot}}, \forall n,\\
		& ~~\textup{tr}(\hat{\bf{G}} ) = 1, \textup{rank}(\hat{\bf{G}})=L, \hat{\bf{G}}\succeq 0.
	\end{aligned}  
\end{equation}
We observe that the only non-convex constraint in problem \eqref{p_s_3} is $\textup{rank}(\hat{\bf{G}})=L$. Therefore, we remove this constraint to obtain a relaxed version of problem \eqref{p_s_3} as follows,
\begin{equation}\label{p_s_4}
	\begin{aligned}
		\min_{\alpha, \hat{\bf{G}}} &~~ \alpha\\
		\textup{s.t.} & ~~ \frac{L_0}{\alpha \lambda_{\min}\left( \bf{H}_n^{\s{H}} \hat{\bf{G}} \bf{H}_n \right)} \leq P_{n}^{\textup{max}}-e_n m_n s^{\textup{tot}}, \forall n,\\
		& ~~\textup{tr}(\hat{\bf{G}} ) = 1, \hat{\bf{G}}\succeq 0.
	\end{aligned}  
  \end{equation}
Then, problem \eqref{p_s_4} can be proved to be a convex problem.
After solving problem \eqref{p_s_4} using a convex solver (e.g., the CVX toolbox in MATLAB) and obtaining the globally optimal solution $\hat{\bf{G}}^*$, we apply the Gaussian randomization algorithm \cite{luo2010semidefinite} to map the solution to a feasible, near-optimal solution for the original non-convex problem.

\addtolength{\topmargin}{0.01in}
\subsection{Long-Term Model Assignment Optimization for $\ca{P}_l$}

The long-term model assignment problem $\ca{P}_l$ is intractable since the optimal $\mathbf{A}^{*}(\bf{m})$ and $\left\lbrace \mathbf{B}_n^{*}(\bf{m})  \right\rbrace$ cannot be derived in closed form. Additionally, the channel state distribution is difficult to obtain in practical wireless systems. To address these challenges, we employ a stochastic SCA algorithm that solves problem $\ca{P}_l$ recursively without requiring prior knowledge of the channel state distribution.
For clearer algorithmic description, we first reformulate the long-term problem $\ca{P}_l$ into an equivalent form as follows,
\begin{equation}
	\vspace{-3pt}
	\begin{aligned}		
		\min_{\bf{m}} &~~ f_0(\bf{m})=\mathbb{E}_{\mathbf{H}}\left[ \textup{MSE}(\mathbf{A}^{*}(\bf{m}),\left\lbrace \mathbf{B}_n^{*}(\bf{m})  \right\rbrace )    \right]  \\
		\textup{s.t.}& ~~ f_1(\bf{m})= s^{\textup{tot}} \textup{diag}(\bf{e} \bf{m}^{\s{T}})  + \frac{L_0}{L} \bf{e}_c\left(\bf{m} \right) \leq \bf{p}^{\textup{max}},\\
		&~~ \bf{m}^{\textsf{T}}\boldsymbol{1}=1, \bf{m}\geq 0,
	\end{aligned}    
\end{equation}
where 
\vspace{-2pt}
$$\bf{e}_c\!\left(\bf{m} \right) \!=\! [\textup{tr}\!\left( \mathbf{B}_1^{*}(\bf{m}) (\mathbf{B}_1^{*}(\bf{m}))^{\s{H}} \right)\!, \ldots, \textup{tr}\!\left( \mathbf{B}_N^{*}(\bf{m}) (\mathbf{B}_N^{*}(\bf{m}))^{\s{H}} \right)]^{\s{T}}\!,$$
$\bf{p}^{\textup{max}}=[{P}_1^{\textup{max}}, \ldots, {P}_N^{\textup{max}}]^{\s{T}}$, and $\bf{e} = [e_1,\ldots,e_N]^{\s{T}}$.
The proposed stochastic SCA algorithm iteratively performs the following steps: First, quadratic surrogate functions are constructed to approximate the non-convex components of the problem.
Then, the resulting convex quadratic approximation is solved, and the long-term model assignment policy is updated based on the solution. The details of these two steps are explained as follows.

\subsubsection{Step 1}

In each iteration $\tau$, the edge server first generates a channel sample $\bf{H}^{\tau}$, and then calculates the short-term transceiver beamformers $\mathbf{A}^{*}(\bf{m}^\tau)$ and $\left\lbrace \mathbf{B}_n^{*}(\bf{m}^\tau)  \right\rbrace$ by solving the short-term problem $\ca{P}_s$.
To address the channel randomness in the objective function $f_0(\bf{m})$, we approximate the expected MSE by using the MSE computed from a specific channel realization. Specifically, in the $\tau$-th iteration, $f_0(\bf{m}^{\tau})$ is evaluated as $\textup{MSE}(\mathbf{A}^{*}(\bf{m}^{\tau}), \left\lbrace \mathbf{B}_n^{*}(\bf{m}^{\tau}) \right\rbrace)$ using the given channel sample $\bf{H}^{\tau}$.
Then, the recursive convex approximations of the original objective function $f_0(\bf{m})$ and the power constraint function $f_1(\bf{m})$ are defined as follows,
 \begin{equation}\label{surrogate_function}
	\begin{aligned}
		\hat{f}_i^{\tau}(\bf{m}) = f_i(\bf{m}^{\tau}) + (\bf{u}_i^{\tau})^{\!\s{T}}\left( \bf{m} \!-\! \bf{m}^{\tau} \right)  + \eta_i \left\|  \bf{m} - \bf{m}^{\tau} \right\|^2,  \\
		\forall i\in \{0,1 \},
	\end{aligned}    
\end{equation}
where $\eta_i$ is a constant that ensures convexity. Furthermore, $\bf{u}_i^{\tau}$ is an approximation of the gradient $\nabla f_i(\bf{m}^{\tau})$, which is updated recursively as follows,
\begin{equation}
	\begin{aligned}
		\bf{u}_i^{\tau} = (1 - \rho^\tau) \bf{u}_i^{\tau-1} + \rho^\tau \nabla_{\bf{m}} f_i(\bf{m}; \mathbf{A}^{*}(\bf{m}^\tau),\left\lbrace \mathbf{B}_n^{*}(\bf{m}^\tau)  \right\rbrace).
	\end{aligned}    
\end{equation}
The sequence $\rho^\tau$ is decreasing in $\tau$, satisfying $\lim_{\tau \rightarrow\infty}\rho^{\tau}=0$, $\sum_{\tau=0}^{\infty}\rho^{\tau}=\infty$, and $\sum_{\tau=0}^{\infty}(\rho^{\tau})^2<\infty$.

\subsubsection{Step 2}

\begin{figure*}[htbp]
	\centering
	\begin{subfigure}{0.328\textwidth}
		\centering
		\includegraphics[width=\linewidth]{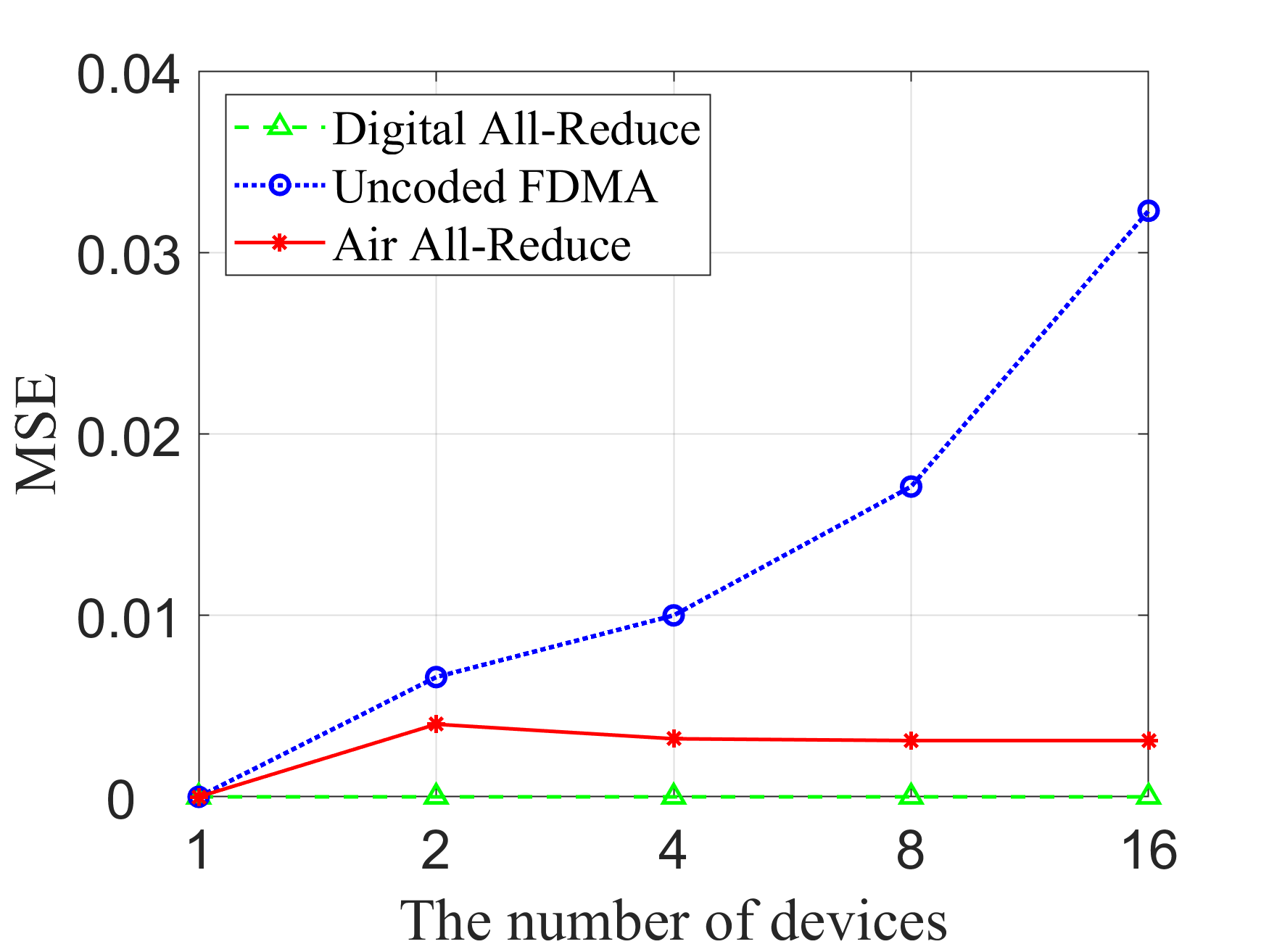}
		\label{fig:subfig1}
	\end{subfigure}
	\begin{subfigure}{0.328\textwidth}
		\centering
		\includegraphics[width=\linewidth]{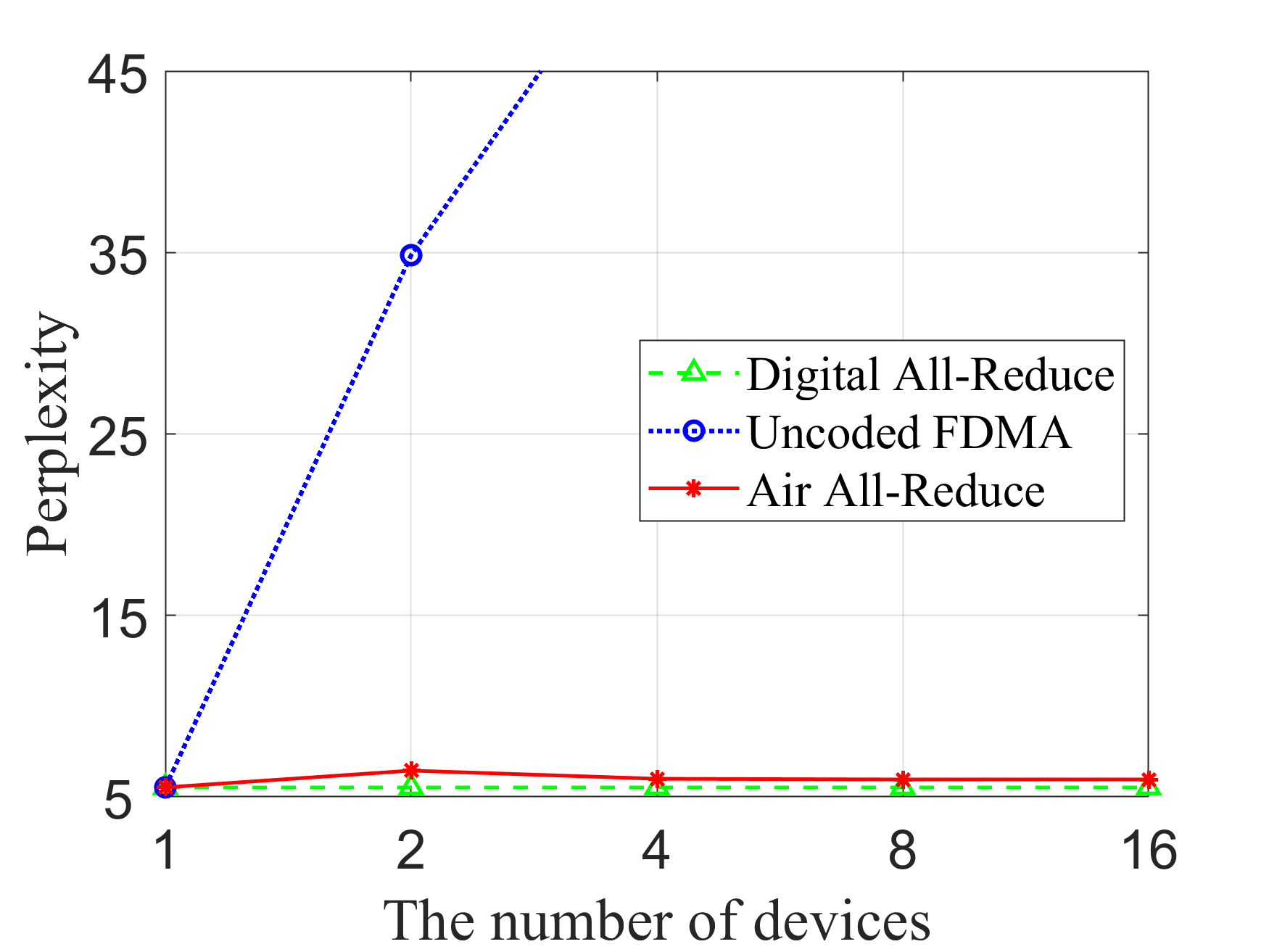}
		\label{fig:subfig2}
	\end{subfigure}
	\begin{subfigure}{0.328\textwidth}
		\centering
		\includegraphics[width=\linewidth]{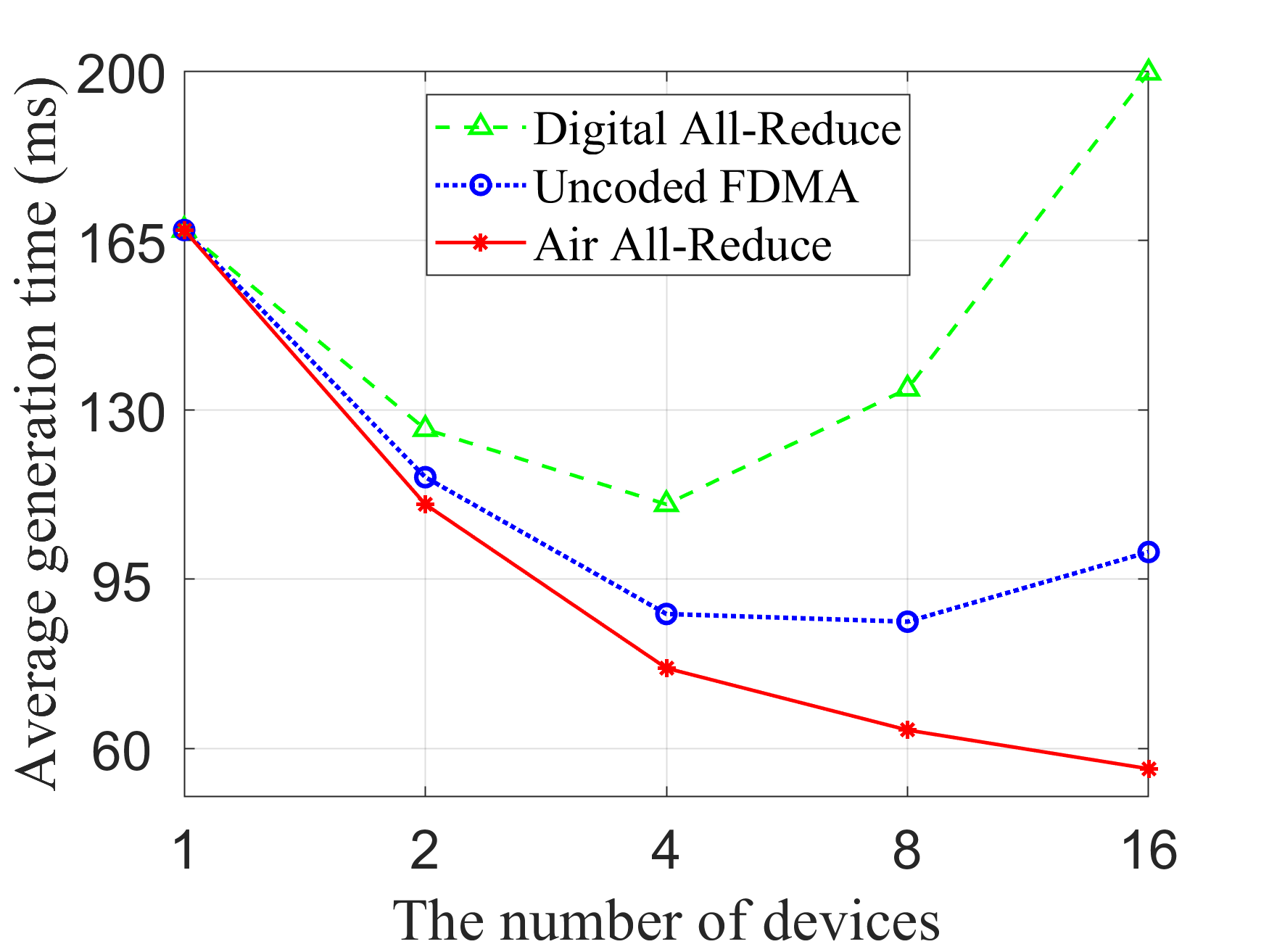}
		\label{fig:subfig3}
	\end{subfigure}
	\vspace{-24pt}
	\caption{The MSE, perplexity, and average generation time versus the number of edge devices}
	\label{fig:three_subfigs}
\end{figure*}

\setlength{\tabcolsep}{3.9pt}
\begin{table*}[t]
	\centering
	\small
	\begin{tabular}{lcccccccccccccccc}
		\toprule
		&  \multicolumn{16}{c}{\textbf{Average generation time per token (ms)}} \\
		\hhline{|~|-|-|-|-|-|-|-|-|-|-|-|-|-|-|-|-|}
		\multirow{1}{*}{\textbf{Model}}  & \multicolumn{4}{c}{\textbf{LLaMA2-7B}} & \multicolumn{4}{c}{\textbf{LLaMA2-13B}} & \multicolumn{4}{c}{\textbf{LLaMA2-70B}} & \multicolumn{4}{c}{\textbf{LLaMA3-70B}}\\
		\multirow{1}{*}{\textbf{Device Number}} &  \textbf{1} & \textbf{2} & \textbf{4} & \textbf{8} &  \textbf{1} & \textbf{2} & \textbf{4} & \textbf{8}& \textbf{1} & \textbf{2} & \textbf{4} & \textbf{8} &  \textbf{1} & \textbf{2} & \textbf{4} & \textbf{8} \\
		\midrule
		Digital All-Reduce & 114.2 & 85.2 & 79.5 & 108.3 & 217.3 & 174.0 & 176.6  &261.4 &\hspace{4.2pt}$\textup{N/A}^{*}$&807.3&729.7&981.6&N/A&893.2&783.8&1033.6\\
		Air All-Reduce & 114.2 & 69.7 & 45.7 & {\textbf{37.8}} & 217.3 & 128.5 & 81.3 & {\textbf{66.4}} & N/A & 660.9 &423.0& {\textbf{354.2}}&N/A&746.8&477.1& {\textbf{406.0}}\\
		\bottomrule
		\multicolumn{5}{c}{*: Not available due to insufficient memory.}&
	\end{tabular}
	\vspace{-3pt}
	\caption{Average generation time for different models across varying device numbers, with the shortest average generation time for each model being highlighted in bold.}
	\vspace{-2pt}
\end{table*}

\begin{algorithm}[tbp]
	\caption{Mixed-Timescale Model Assignment and Transceiver Optimization Algorithm} \label{algorithm1}
	\textbf{Initialize:} Model assignment policy $\bf{m}^{0}$ and iteration index $\tau=0$\;
	\textbf{\mbox{Step 1 (long-term model assignment optimization at the} beginning of inference task)} \\
	
	\mbox{Obtain a channel sample $\bf{H}^{\tau}$ and calculate the short-term} \mbox{transceiver beamformers $\!\mathbf{A}^{\!*}(\bf{m}^\tau),\left\lbrace \mathbf{B}_n^{*}(\bf{m}^\tau)  \right\rbrace$ by solving} the short-term problem $\ca{P}_s$\;
	
	Update the surrogate function $\hat{f}_i^{\tau}(\bf{m})$ according to \eqref{surrogate_function}\;
	
	\mbox{Solve problem \eqref{sca_problem} to obtain the optimal $\hat{\bf{m}}^{\tau}$ and update} $\bf{m}^{\tau}$ according to \eqref{z_update_rule}\;
	\mbox{Let $\tau = \tau \!+ \!1$ and return to Step 1. Repeat the above steps} until convergence\;
	
	\textbf{\mbox{Step 2 (short-term transceiver optimization at each all-} reduce step):} \\
	\mbox{Obtain the channel condition $\bf{H}$ and apply the short-term} \mbox{algorithm to solve the optimal transceiver beamformers} with the determined model assignment policy $\bf{m}$.
	
\end{algorithm}

After obtaining the convex approximation of the objective function and the constraint function, we formulate a convex approximation of the original problem as follows:
\begin{equation}\label{sca_problem}
	\begin{aligned}
		\hat{\bf{m}}^{\tau}=\min_{\bf{m}} &~~ \hat{f}_0^\tau(\bf{m}) \\
		\textup{s.t.}& ~~ \hat{f}_1^\tau(\bf{m}) \leq \bf{p}^{\textup{max}},\\
		&~~ \bf{m}^{\textsf{T}}\boldsymbol{1}=1, \bf{m}\geq 0.
	\end{aligned}    
\end{equation}
After solving for $\hat{\bf{m}}^{\tau}$, the model assignment policy is updated as follows,
\begin{equation}\label{z_update_rule}
	\begin{aligned}
		\bf{m}^{\tau+1} = (1 - \gamma^{\tau}) \bf{m}^{\tau} + \gamma^{\tau} \hat{\bf{m}}^{\tau},
	\end{aligned}    
\end{equation}
where $\gamma^{\tau} \in (0, 1)$ satisfies $\lim_{\tau \rightarrow\infty}\gamma^{\tau}=0$, $\sum_{\tau=0}^{\infty}\gamma^{\tau}=\infty$, and $\sum_{\tau=0}^{\infty}(\gamma^{\tau})^2<\infty$.

The above two steps iterate until convergence, and the overall process is outlined in Algorithm 1. Moreover, the convergence of the stochastic SCA algorithm has been rigorously analyzed in \cite{liu2018stochastic}.

\section{Simulation Results}


\subsection{Simulation Setups}

\subsubsection{LLM Inference Model Setting}

In the simulations, $N$ virtual machines (VMs) are set up on a single desktop, with each VM simulating a distinct edge device.
For evaluation, we utilize the LLaMA2 and LLaMA3 models, along with the WikiText-2 dataset.
The primary performance metric for inference accuracy is perplexity, a widely recognized measure of a LLM's capability to predict the next word in a sequence. It is defined mathematically as follows,
\vspace{-1.7pt}
\begin{equation} 
	\vspace{-1.7pt}
	\begin{aligned} 
		&\mathrm{Perplexity} =\exp \!\left(\!-\frac{1}{L_{\textup{txt}}} \sum_{k=1}^{L_{\textup{txt}}} \log \b{P}{}\left(w_k \!\mid\! w_1, \ldots, w_{k-1}\right)\right)\!,
		\end{aligned}
\end{equation}
where $\b{P}{}\left(w_k \mid w_1, \ldots, w_{k-1}\right)$ is the model's predicted probability for the next word, and $L_{\textup{txt}}$ is the text length.
Lower perplexity values indicate better inference performance, reflecting the model’s accuracy in generating subsequent tokens.

\subsubsection{Communication Model Setting}

The antenna number at the edge server is $N_r=20$, and each edge device has $N_t=4$ antennas. The bandwidth between the edge server and edge devices is $B = 10$ MHz.
The uplink channels are assumed to be independent and identically distributed (i.i.d.) Rician fading, modeled as i.i.d. complex Gaussian random variables with non-zero mean $\mu = 1$ and variance $\sigma^2 = 1$. Moreover, the maximum power budget is set as $P_{n}^{\textup{max}}= 1$ and the noise variance at the edge server is assumed to be 1.

\vspace{-2pt}
\subsection{Performance Evaluation}
\vspace{-2pt}

We compare the performance of the proposed air all-reduce approach with the following two benchmark schemes.

\begin{itemize}
	\item \textbf{Digital All-Reduce}: All devices upload intermediate layer outputs using a traditional broadband digital multi-access scheme, with each transmitted symbol quantized to $Q=8$ bits. To prevent multi-user interference, orthogonal frequency division multiple-access (OFDMA) is used, assigning each sub-channel to one device.
	\item \textbf{Uncoded FDMA}: This scheme similarly employs the OFDMA technique, with each device occupying a dedicated sub-channel to upload intermediate layer outputs in an uncoded analog manner.
	
\end{itemize}

In Fig. \ref{fig:three_subfigs}, we employ the LLaMA3 model with 8 billion parameters to compare the inference performance of different algorithms, encompassing three key performance metrics: transmission MSE, perplexity, and average generation time.
In Fig. \ref{fig:three_subfigs}(a), the proposed air all-reduce approach consistently achieves low MSE across all device counts, significantly outperforming the uncoded FDMA scheme, which exhibits a near-linear increase in MSE as the number of devices grows. The digital all-reduce method achieves near-zero MSE, but it will incur higher time costs (discussed later).
In Fig. \ref{fig:three_subfigs}(b), perplexity follows the same trend as the transmission MSE. The air all-reduce method maintains stable, low perplexity across all device configurations, while the perplexity of uncoded FDMA rises sharply with more devices. Digital all-reduce performs similarly to air all-reduce, maintaining low perplexity.
For the average generation time as shown in Fig. \ref{fig:three_subfigs}(c), air all-reduce offers the lowest latency among the three methods, particularly as the number of devices increases. Digital all-reduce has significant latency increases with more devices, while uncoded FDMA shows moderate latency improvements but still lags behind air all-reduce. Overall, air all-reduce strikes a balance between low latency and high inference accuracy, outperforming both benchmarks in practical wireless scenarios.

To further validate the effectiveness of the proposed algorithm, we conduct additional experiments using larger models, including LLaMA2 with 7, 13, and 70 billion parameters, and LLaMA3 with 70 billion parameters. In Table I, it can be observed that the proposed air all-reduce method consistently outperformed digital all-reduce in terms of inference speed. Across various device configurations, air all-reduce achieves up to 3x faster generation time, demonstrating its significant advantages for distributed LLM inference, especially with large-scale models.

\section{Conclusion}

In this paper, we investigated communication-efficient distributed on-device LLM inference over wireless networks.
We proposed an over-the-air computation approach to accelerate the frequent all-reduce operations required in the tensor parallelism-based distributed inference.
To minimize the average transmission MSE, we formulated a joint model assignment and transceiver design problem, which can be derived as a mixed-timescale stochastic non-convex optimization.
We developed a mixed-timescale algorithm by leveraging the SDR and stochastic SCA methods to address this problem.
Simulation results demonstrate that the proposed approach significantly reduces inference latency while improving accuracy, making distributed on-device LLM inference feasible for resource-constrained edge devices.


\bibliographystyle{ieeetr}
\bibliography{IEEEabrv,refs}


\end{document}